\documentclass{article}
\usepackage[utf8]{inputenc} %unicode support
\usepackage{amsmath}
\usepackage{amssymb}
\usepackage{amsthm}
\usepackage{graphicx}
\usepackage{natbib}
\usepackage{authblk}

\newtheorem*{NMA2}{Two-region NMA theorem}
\newtheorem*{NMA2B}{Two-ball NMA theorem}
\newtheorem*{NMA}{Corollary: General NMA theorem}
\newtheorem*{USA}{Unique source assignment theorem}
\newtheorem*{Kellogg}{Kelloggs theorem 10.5}

\newcommand{\rev}[1]{ #1 }

\title{A  uniqueness theorem in potential theory with implications for tomography-assisted inversion}

\author[1]{Karl Fabian}
\author[2]{Lennart V. de Groot}
\affil[1]{  Geological Survey of Norway, Leiv Eirikssons vei~39, \emph{7491}~Trondheim, Norway}
\affil[2]{Paleomagnetic laboratory Fort Hoofddijk, Department of Earth Sciences, Utrecht University, Budapestlaan 17, 3584~CD~Utrecht, The~Netherlands.}

\date{\today}

\begin{document}

\label{firstpage}

\maketitle

\begin{abstract}
Inversion of potential field data is central for remote sensing in physics, geophysics, neuroscience
and medical imaging.
Potential-field inversion results are improved by including constraints from independent measurements, like tomographic source localization, but so far
no mathematical theorem  guarantees  that such prior information can yield uniqueness of the achieved assignment.
Standard potential theory is used here to prove a uniqueness theorem which completely characterizes the mathematical background of source-localized inversion. It guarantees for an astonishingly large class of source localizations that it is possible by potential field measurements on a surface to differentiate between signals from prescribed source regions. The well-known general non-uniqueness of potential field inversion only prevents that the source distribution  within the individual regions can be uniquely recovered.
This result enables large scale surface scanning to reconstruct reliable magnetization directions of localized magnetic particles and provides an incentive to improve scanning methods for paleomagnetic applications.
\end{abstract}

\section{Introduction}
It is long known that  a charge distribution inside a sphere cannot be uniquely reconstructed from
potential field measurements on or outside this sphere, because every charge distribution inside can be replaced by an equivalent surface charge distribution creating the same outside potential \citep{Kellogg:1929}.
\rev{
When  inverting magnetic  field surface measurements, all mathematical approaches make substantial additional
assumptions about the source magnetization to achieve   useful reconstructions \citep[see e.g.][]{Zhdanov:2015,Baratchart:2013}.
}
To still infer localized information in spite of this non-uniqueness, we previously  suggested  to constrain the  source regions inside a region $\Omega$ by additional tomographic information \citep{DeGroot:2018}.
\rev{
The corresponding inversion algorithm turned out to be extremely successful and efficient which seemed to deserve a mathematical underpinning.
This led to  a new type of inversion problem, namely to assign parts of the total measured signal to  charge distributions inside  regions $P_1,\ldots,P_N$ that beforehand have been tomographically outlined.
}
Is it now still possible that some  non-zero charge distribution, for example inside particles $P_1, P_2, P_4, P_5$ in Fig.~\ref{Sketch-1}a,  creates exactly the same measurement signal as another charge distribution inside the omitted particle $P_3$?
Here it is shown that  this is not the case for regions $P_1,\ldots,P_N$ which are topologically separated in a sense specified below. Accordingly, a potential field measurement at the surface of $\Omega$
can be uniquely decomposed into  signals from  such individual, preassigned source regions. By that, the inevitable non-uniqueness of potential-field  inversion turns out to be  completely constrained to the uncertainty of the internal source distribution within these individual regions.

\begin{figure}
\centering
 \includegraphics[width=80mm]{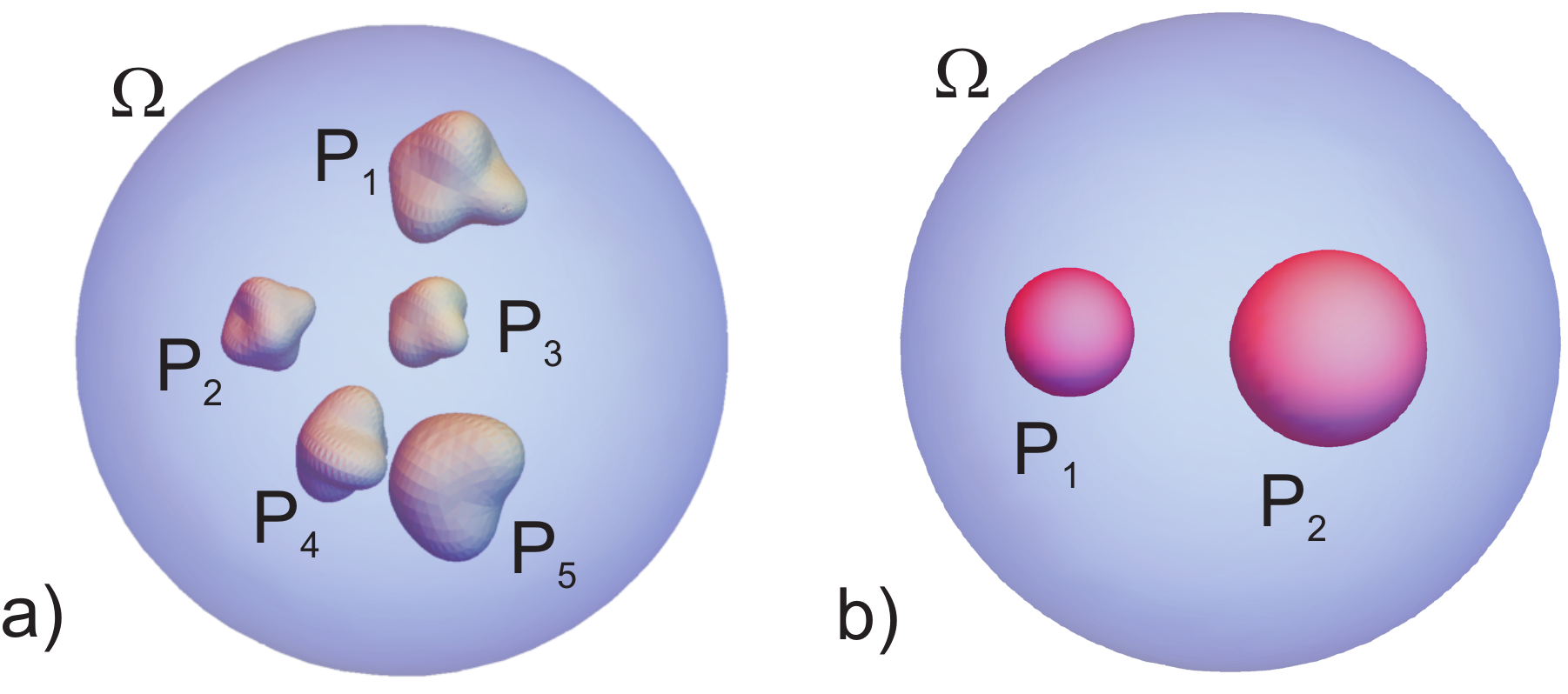}
\caption{Geometric situation a) in the general case, and b) for the simplified model.  }\label{Sketch-1}
\end{figure}

To show this result, it is first proved  that there is no non-zero charge distribution in one region, that annihilates the signal of a charge distribution in another, topologically separated region. This is the content of the No-Mutual-Annihilator theorem in section ~\ref{NMA}.

From that the main theorem on unique source assignment in section~\ref{USA} follows directly by the linearity of the von Neumann boundary value problem for the Poisson equation. Therefore, the main mathematical content is encapsulated in the  two-region NMA theorem, that can be regarded as a \rev{ far reaching} generalization of a theorem of Gauss about separating the internal and external components of the geomagnetic field \citep{Gauss:1877,Backus:1996}.\rev{  It} essentially relies on the fact that harmonic functions are analytic and can be uniquely analytically continued on simply connected open sets \citep[Theorem 1.27]{Axler:2001}.

\section{The No-Mutual-Annihilator theorem}\label{NMA}

Let $\Omega \subset \mathbb{R}^3$ be open and $\partial \Omega$ a nonempty, smooth compact manifold.
For a set   $G$ with   $\overline{G}\subset\Omega$  the (Neumann) annihilator of $G$ in $\partial \Omega$ is defined as
\begin{eqnarray*}
&&{\rm Ann} (G) := \left\{ \rho \in L^1(G): {\rm supp}\,\rho \subset \mathring{G},\right.
\\~&~&\left. \exists \Phi\in C^2(\Omega)\cap C^1(\overline{\Omega}):
\Delta\,\Phi = \rho {~~\rm and~~} \frac{\partial \Phi}{\partial n}  = 0 {\rm~on~}\partial \Omega\right\}.
\end{eqnarray*}

Physically, ${\rm Ann} (G)$ represents the vector space of all possible charge distributions inside the region $G$ which create no measurement signal on the boundary $\partial \Omega$. Because the measurement signal is the normal derivative of the potential field, the potential itself is only defined up to a globally constant summand, and in the following this constant is chosen such that the analytic continuation of $\Phi$ to $\mathbf{R}^3$ vanishes at infinity. The corresponding potentials are called zero-gauged.

 $N$ pairwise disjoint compact sets $P_1,\ldots,P_N$ with $P_i\subset \Omega$
have the {\em no-mutual-annihilator} (NMA) property if
$$ {\rm Ann}(\bigcup \limits_{i=1}^N P_i) ~=~\bigoplus\limits_{i=1}^N {\rm Ann} (P_i). $$

In the above equation the ''$\supset$'' inclusion is always true, because any element of the vector space spanned by the annihilators of the $P_i$
is an annihilator of the union $\bigcup \limits_{i=1}^N P_i$.
The other inclusion ''$\subset$'' in the NMA property implies, that it is impossible to have a charge distribution $\rho$ within the region $\bigcup \limits_{i=1}^N P_i$ which generates a zero signal on the boundary, such that if the charge distribution is set to zero in some, but not all, of the $P_i$, the resulting boundary signal is not zero.\\

An example of two sets which do not have the NMA property are two nested balls
$P_1=B(r)$ and $P_2=B(R)\backslash B(r)$ for $0<r<R$. A well-known annihilator in this case are constant non-zero charge distributions of opposite sign such that the integral over $B(R)$ is zero \citep{Zhdanov:2015}. Setting the charge distribution in one of $P_1,P_2$ to zero clearly generates a non-zero field on $\partial \Omega$.

Intuitively it appears plausible that two point charges inside a sphere, which lie far apart from each other, but close to the surface of the sphere do have the NMA property. At least if the charge distribution inside one of them has a non-zero total charge, then the other must have the opposite total charge to annihilate the field at large distance, but at small distance  on the surface $\partial \Omega$ these charges cannot cancel each other.

It is also known that the annihilator sets ${\rm Ann} (G)$  for $\overline{G}\subset \Omega$ are large.
For   star-shaped  $G$,
any charge distribution $\rho\in L^1(G)$   which for all harmonic functions $h\in C^2(\Omega): \Delta\,h\,=\,0$ fulfills
$$  \int \limits_G h(r)\,\rho(r)\,dV \,=\, 0 $$
 generates no field on $\partial \Omega$, such that $\rho\in {\rm Ann} (G)$ \citep{Zhdanov:2015}.
This apparently bleak state of affairs with respect to unique-inversion results is emphasized by the fact that \citep{Zhdanov:2015} reports as the best result so far that if  a gravity field is generated by a star-shaped body  of constant density $\rho(r) = \rho_0$, the gravity inverse problem has a unique solution \citep{Novikov:1938}.

It therefore may appear incredible that a far-reaching uniqueness result, as claimed above, is not in conflict with the known non-uniqueness results. We will now show that it is mathematically feasible. To make the proof easier to follow, it is first shown under relatively weak topological conditions that two disjoint regions have the NMA property. By induction this is then generalized to a finite number of $N$ regions. The essential property to avoid non-uniqueness of the inversion is that the complement of these finite regions is a simply connected region where analytical continuation is uniquely possible. Thereby uniqueness of the inversion is closely linked to uniqueness of analytical continuation.

\begin{NMA2}
Let $\Omega \subset \mathbf{R}^3$ be open and $\partial \Omega$ a smooth compact manifold and
 $P_1, P_2  \subset \Omega$ be disjoint compact sets, such that $\mathbb{R}^3\backslash P_1$, $\mathbb{R}^3\backslash  P_2$, and $\mathbb{R}^3\backslash (P_1 \cup P_2)$ are simply connected   then
$P_1$ and $P_2$ have the {\em No-Mutual-Annihilator} property with respect to $\Omega$.\\
\end{NMA2}

\begin{figure*}
\centering
 \includegraphics[width=120mm]{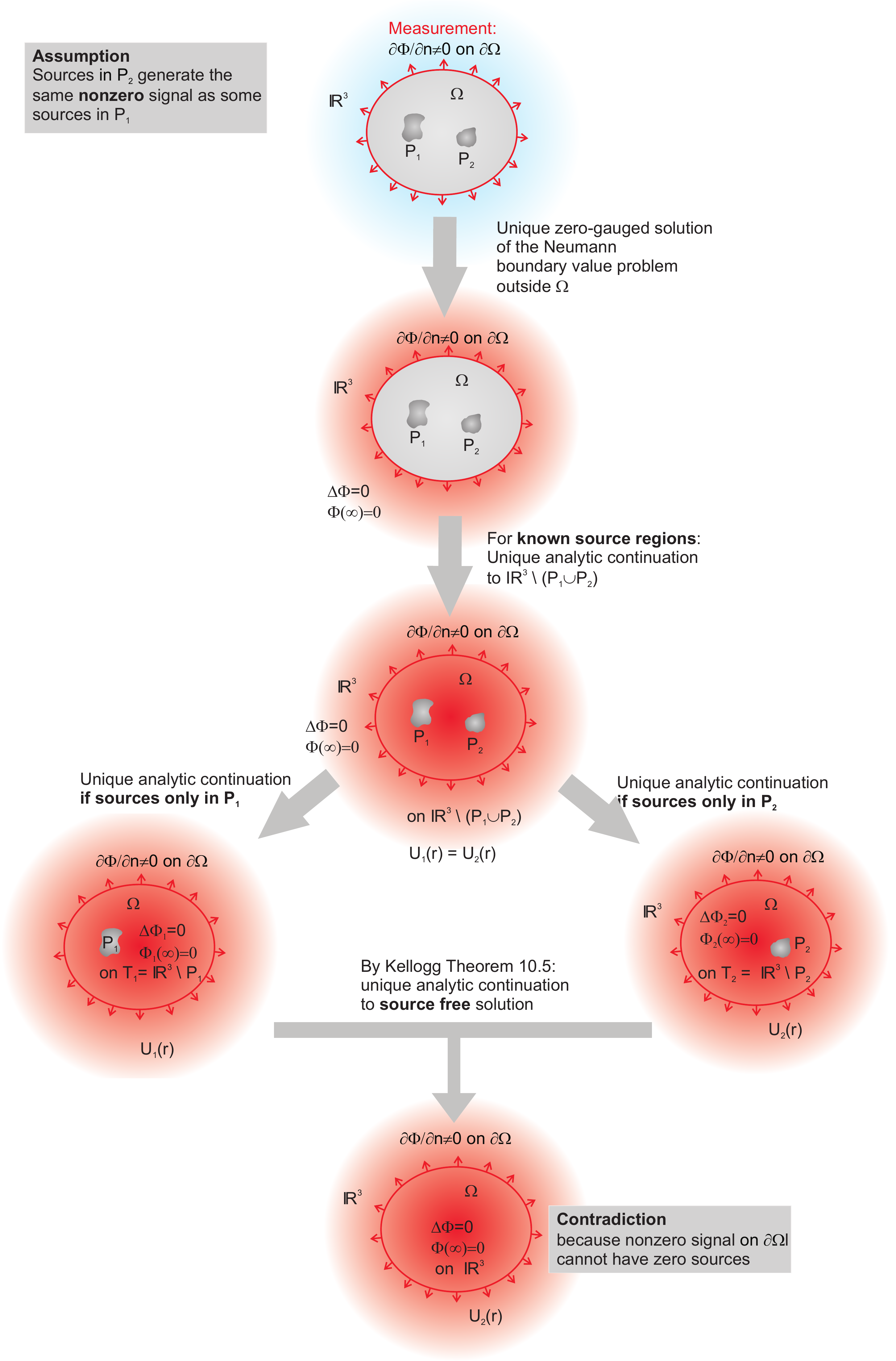}
\caption{\rev{Overview of the proof of the two-region NMA theorem. The assumption that sources in region $P_1$ generate the same nonzero field on the surface $\partial \Omega$ as sources inside region $P_2$ leads to a contradiction
if $ T_1$ and $T_2$ are simply connected.}
}\label{Sketch-2}
\end{figure*}

\begin{proof}
We derive a contradiction from the assumption that there exists a mutual annihilator
\rev{
$$\rho \in {\rm Ann}( P_1 \cup P_2) \backslash  ( {\rm Ann} (P_1) \oplus {\rm Ann} (P_2)).$$
}
By definition, then there are two nonzero  functions $\rho_1,\rho_2 \in  L^1(\Omega)$ with  ${\rm supp}\,\rho_1 \subset P_1$, ${\rm supp}\,\rho_2 \subset P_2$, such that
$$ \rho = \rho_1-\rho_2,$$
and the non-zero normal derivatives of their potentials $\frac{\partial \Phi_1}{\partial n}, \frac{\partial \Phi_2}{\partial n}$ are identical on $\partial\Omega$.
Now recall that  the solution of the Neumann problem for harmonic functions is unique for zero-gauged potentials \citep[Theorem 8.4]{Kellogg:1929}, by which $\Phi_1=\Phi_2$ on $\mathbb{R}^3\backslash \Omega $, where
a potential $U$ is called zero-gauged, if
$$\lim \limits_{||x|| \to \infty} U(x) ~=~0.$$
We  now conjure up a bit of mathematical magic in form of  Theorem 10.5 in \citep{Kellogg:1929} which essentially encapsulates Gauss theorem of separation of sources.
By assumption, the sets  $T_1:=\mathbb{R}^3 \backslash P_1$
and $T_2 := \mathbb{R}^3\backslash P_2$ are simply connected and open and overlap   on the simply connected set $\mathbb{R}^3\backslash ( P_1 \cup P_2)$.
By analytic continuation
  on the simply connected open sets $T_1$ and $T_2$  \citep[Theorem 1.27]{Axler:2001} there is a  unique harmonic
function $U_1$ on $T_1$ with $U_1= \Phi_1$ on $\mathbb{R}^3\backslash \Omega$, and a unique $U_2$ on $T_2$ with $U_2= \Phi_2$ on $\mathbb{R}^3\backslash \Omega$.
By \citep[Theorem 10.5 ]{Kellogg:1929}, there now also is a unique harmonic function $U$ on
$\mathbb{R}^3$ with $U=U_1$ on $T_1$ and $U=U_2$ on $T_2$. This implies that $U$ solves the zero-gauged Neumann problem $\Delta U=0$ on $\mathbb{R}^3$  with boundary condition $\frac{\partial U}{\partial n}=\frac{\partial \Phi_1}{\partial n}$ on $\partial \Omega$. Because the unique  zero-gauged potential with $\Delta U=0$ on $\mathbb{R}^3$ is $U=0$,
it follows that $\rho_1=\rho_2=0$ which contradicts the assumption.
%$\blacksquare$
\end{proof}

Because the above proof is quite mathematical in nature, in the supplementary information the  special case of a two-ball NMA theorem, in which $P_{1,2}$ are disjoint balls as in Fig.~\ref{Sketch-1}b,  is  proved by directly applying   Gauss theorem of separation of sources. This may help to acquire a physical understanding of the strength and limitations of the result, and  may also lend more credulity to the derivation above. In the next step the result of the two-region NMA theorem is extended to arbitrary numbers of regions by induction.

\begin{NMA}
Let $\Omega \subset \mathbf{R}^3$ be open and $\partial \Omega$ a smooth compact manifold.
For a natural number $N\geq 1$ let
 $P_1,\ldots,P_N  \subset \Omega$  be pairwise  disjoint compact sets, such that
 $\mathbb{R}^3\backslash P_k$ and  $\mathbb{R}^3\backslash \bigcup \limits_{i=1}^k P_i$
 are simply connected for all $k=1,\ldots,N$.
 Then
the $P_i$ have the {\em No-Mutual-Annihilator} property with respect to $\Omega$.
\end{NMA}

\begin{proof}
For $N=1$ there is nothing to prove. Assume that $N>1$ and that the corollary is true for $N-1$. Define the sets $P_1'=\bigcup \limits_{i=1}^{N-1} P_i$ and $P_2'=P_N$.
The assumptions on the $P_k$ imply that $P_1'$ and $P_2'$ fulfill the conditions
to apply the two-region NMA theorem, whereby
$P_1'$ and $P_2'$ have the {\em No-Mutual-Annihilator} property with respect to $\Omega$
which implies
\rev{
$$ {\rm Ann}(\bigcup \limits_{i=1}^N P_i) ~=~{\rm Ann}(\bigcup \limits_{i=1}^{N-1} P_i)
\oplus  {\rm Ann} (P_N). $$
Because the corollary is true for $N-1$ and $P_1,\ldots,P_{N-1}$ fulfill the conditions for its application we have by induction
$$ {\rm Ann}(\bigcup \limits_{i=1}^{N-1} P_i) ~=~\bigoplus \limits_{i=1}^{N-1} {\rm Ann} (P_i). $$
}
Substituting this in the above equation proves the corollary.
%$\blacksquare$
\end{proof}

\section{Unique source assignment}\label{USA}
The previous two theorems provide all prerequisites to formulate the main result of this article:
\begin{USA}
Let $\Omega \subset \mathbb{R}^3$ be open, simply connected, and $\partial \Omega$ a smooth compact manifold. Assume that
$P_1,\ldots,P_N  \subset \Omega$  are pairwise  disjoint compact sets such that
$\mathbb{R}^3\backslash P_k$ and  $\mathbb{R}^3\backslash \bigcup \limits_{i=1}^k P_i$
 are simply connected for all $k=1,\ldots,N$.
If the sources of the  zero-gauged potential  $\Phi$ have compact support on $\bigcup \limits_{k=1}^N P_k $, then  $\frac{\partial \Phi}{\partial n}$ on $\partial \Omega$ uniquely determines zero-gauged  potentials
$\Phi_1,\ldots,\Phi_N$,
such that $\Phi_i$ is harmonic on $\mathbb{R}^3 \backslash \bigcup \limits_{k\neq i} P_k $, which implies that it has no sources outside $P_i$, and
$$\frac{\partial \Phi}{\partial n}~=~ \sum  \limits_{i=1}^N \frac{\partial \Phi_i}{\partial n} ~~{\rm on}~~\partial \Omega.$$
\end{USA}

\begin{proof}
% Note that for the compact pairwise disjoint sets $P_1,\ldots,P_N $ one has   $\Omega \backslash P_i \cup \Omega \backslash P_j \,=\, \Omega $ for $i\neq j$, whereby all sets are open.\\
Because the source of $\Phi$ is a charge distribution $\rho$ in $\bigcup \limits_{k=1}^N P_k $  there exist zero-gauged harmonic potentials
$\Phi_1,\ldots,\Phi_N$ with the required properties, namely those generated by the local charge distributions $\rho_k=\rho|_{P_k}$.\\

Uniqueness  is now shown by the general NMA theorem.
Take any charge distribution  $\rho'$ in $\bigcup \limits_{k=1}^N P_k $
with zero-gauged potentials $\Psi_1,\ldots,\Psi_N$,
 such that $\Psi_i$ is harmonic on $\mathbb{R}^3 \backslash \bigcup \limits_{k\neq i} P_k $ and
 $$\frac{\partial \Phi}{\partial n}~=~ \sum  \limits_{i=1}^N \frac{\partial \Psi_i}{\partial n} ~~{\rm on}~~\partial \Omega.$$
 Then define $\Gamma_i= \Phi_i-\Psi_i$ such that $\Gamma$ with

 $$\Gamma:= \sum  \limits_{i=1}^{N}  \Gamma_i   ~=~0~~{\rm on}~~\mathbb{R}^3\backslash \Omega,~~{\rm and}~~\frac{\partial \Gamma}{\partial n}~=~ 0 ~~{\rm on}~~\partial \Omega,$$
is the zero-gauged potential
 from the source distribution $\rho-\rho'$, which thereby is a member of
 \rev{
$$ {\rm Ann}(\bigcup \limits_{i=1}^NP_i) ~=~\bigoplus \limits_{i=1}^N {\rm Ann} (P_i). $$
}
The equality is due to the general NMA theorem and its right hand side implies that
$ \Gamma_i   ~=~0$, or $\Phi_i=\Psi_i$ for $i=1,\ldots,N$. Thus the zero-gauged $\Phi_i$ are uniquely determined by $\frac{\partial \Phi}{\partial n}$ on $\partial \Omega$.
%$\blacksquare$
\end{proof}

\subsection{Unique source assignment is well-posed}
When denoting by $H_0(\mathbb{R}^3\backslash P)$ the space of harmonic, zero-gauged functions outside a compact region $P$, the  linear operator for solving the inverse problem
$$A:  H_0(\mathbb{R}^3\backslash(P_1\cup P_2))\to  H_0(\mathbb{R}^3\backslash P_1 ),~~\Phi\mapsto~\Phi_1 $$
has the nullspace   $H_0(\mathbb{R}^3\backslash P_2 )$ which is closed
in $H_0(\mathbb{R}^3\backslash(P_1\cup P_2))$, whereby $A$ is continuous \citep[theorem 1.18]{Rudin:1991}. Accordingly the source assignment problem a) has a solution, b) this solution is unique, and c) the operator that maps the measurement to the solution is continuous,
\rev{which by Hadamard's definition \citep{Zhdanov:2015} implies that the inversion  is a well-posed problem. In case of sufficiently dense data and low signal-to-noise ratio the inverse problem therefore can be expected to be solvable in a stable and robust way. As with any inverse problem, the numerical inversion  can still be ill-conditioned, for example in cases where the the discretization is too coarse or the signal-to-noise ratio is low.
}
\section{Consequences}
This new theorem
provides a clear and astoundingly general condition for when it is theoretically possible to uniquely assign potential field signals to source regions. To give a intuitive argument why
this kind of theorem can exist, consider  the simple case when $\Omega$ and all $P_k$ are balls. The theorem now  guarantees
 that from the spherical harmonic expansion of the field on $\partial \Omega$ all individual  spherical harmonic expansions on the $\partial P_k$ are uniquely determined.
Thus the coefficients of one countably infinite basis of an harmonic function space uniquely define $N$  countably infinite coefficient sets
on $N$ infinite bases, which is no contradiction in analogy to the Hilbert-hotel paradox \citep{Hilbert:1924}.
% That any basis of the possible charge distributions in $L^1(P_k)$ is uncountable   explains the overwhelming non-uniqueness of the charge reconstruction problem, and  shows that this uncountable degree of non-uniqueness is retained by the above unique source assignment theorem, although it is confined to the individual regions.\\

Unique source assignment is significant in geophysics for gravimetric, or aeromagnetic  interpretation, when combined with tomographic methods like seismic imaging. It also lies the foundation for reading three-dimensional magnetic storage media.
\rev{
In rock-magnetism, after the pioneering work of \citet{Egli:2000}, different magnetic surface scanning techniques are increasingly used to infer magnetization sources and  magnetization structure inside rocks \citep[e.g.][]{Uehara:2007,Hankard:2009,Usui:2012,Lima:2013,Glenn:2017}.
In this context, the unique source-assignment theorem enables paleomagnetic reconstruction from natural particle ensembles \citep{DeGroot:2018}, because it establishes that individual dipole moments from a large number of magnetic particles in a non-magnetic matrix that are localized  by density tomography (micro-CT) can be uniquely recovered from surface magnetic field measurements.
In  \citet{DeGroot:2018} uniqueness of dipole reconstruction is individually certified by showing that for some specific set of $K$ magnetic particles found by density tomography one can find $3 K$ surface measurements such that the a $3 K\times 3 K$-matrix of the forward calculation is invertible. This proves that only a unique set of dipoles can explain the measurement. The result proven here is much more general in that it asserts, that no two different sets of multipole expansions originating from the  particles can lead to the same surface signal.
The induction proof of the unique source assignment theorem even indicates a divide-and-conquer type strategy for algorithmic implementation of an inverse reconstruction.
}

When scanning a sample in its natural-remanent magnetization state, and again after applying standard paleomagnetic stepwise demagnetization procedures, the resultant demagnetization data set can be  studied on an individual particle level to identify stable and unaltered remanence carriers.
By selecting  only optimally preserved and stable remanence carriers from a large collection of measured particles,
 reliable statistical average paleomagnetic directions or NRM intensities can be calculated for terrestrial or extraterrestrial rocks that due to unresolvable noise currently could not be used as  recorders of their magnetic history.

Further potential application areas of  unique source assignment  theorems are for example \rev{ inversion problems in EEG (electroencephalography), MEG (magnetoencephalography),  or ECG(electrocardiography), where it might enable to uniquely assign externally measured potential field  signals to previously determined brain or heart regions }
\citep{Baillet:2001,Michel:2004,Grech:2008,Michel:2012,Huster:2012}.
Empirical inversion techniques that now use numerical and statistical approaches to assess the reliability of their results \citep{Friston:2008,Castano-Candamil:2015} may profit from  unique source assignment to prior known regions.

What essentially remains impossible   is to assign signals to source regions which lie inside other source regions, \rev{like the nested balls described in section~\ref{NMA}. These cases are excluded, because they do not fulfill the condition of simple connectivity of $\mathbb{R}^3\backslash P_k$ for all $k$, which makes analytic continuation impossible.
The fact that this appears to be the only obstruction to unique reconstruction} provides a new incentive and direction to study potential field measurement techniques in combination with {\em a priori} source localization  to recover
a maximum of information about the spherical harmonic expansion of the individual source regions.

\section*{Acknowledgments}
We wish to thank M. Zhdanov (University of Utah), M. Kunze (Universit\"at zu K\"oln) and R. Egli (ZAMG, Vienna) for helpful comments on an earlier version of the manuscript.

\bibliographystyle{plainnat}
\bibliography{lit}

\newpage

\section*{Supplementary information}

\begin{Kellogg}
If $T_1$ and $T_2$ are two domains with common points, and
if $U_1$ is harmonic  in $T_1$ and $U_2$ in $T_2$, these functions coinciding at the
common points of $T_1$ and $T_2$, then they define a single function, harmonic
in the domain $T$ consisting of all points of $T_1$ and $T_2$.\citep{Kellogg:1929}
\end{Kellogg}

\begin{NMA2B}
Let $\Omega \subset \mathbf{R}^3$ be open and $\partial \Omega$ a smooth compact manifold and
 $P_1, P_2  \subset \Omega$ be disjoint balls, then
$P_1$ and $P_2$ have the {\em No-Mutual-Annihilator} property with respect to $\Omega$.
\end{NMA2B}

\begin{proof}
If there exists a mutual annihilator
$$\rho \in {\rm Ann}( P_1 \cup P_2) \backslash  ( {\rm Ann} (P_1) \oplus {\rm Ann} (P_2)),$$
then there are two nonzero  functions $\rho_1,\rho_2 \in  L_1(\Omega)$ with  ${\rm supp}\,\rho_1 \subset P_1$, ${\rm supp}\,\rho_2 \subset P_2$, and  $ \rho = \rho_1-\rho_2$, such that the non-zero normal derivatives of their potentials $\frac{\partial \Phi_1}{\partial n}, \frac{\partial \Phi_2}{\partial n}$ are identical on $\partial\Omega$. Because  the solution of the Neumann problem for zero-gauged harmonic functions is unique, $\Phi_1=\Phi_2$ on $\mathbf{R}^3\backslash \Omega $.
Because $P_1,P_2$ are disjoint
 $\mathbb{R}^3\backslash \overline{P_1 \cup   P_2}$ is an open simply connected set and the harmonic functions $\Phi_1,\Phi_2$
are defined on $\mathbb{R}^3\backslash \overline{P_1 \cup   P_2}$, and equal on the nonempty open set $\mathbb{R}^3\backslash \overline{\Omega}$.
Because every harmonic function is analytic, this implies $\Phi_1~=~\Phi_2$ on
$\mathbb{R}^3\backslash \overline{P_1 \cup   P_2}$\citep[theorem 1.27]{Axler:2001}\\
For the potential $\Phi_1$ all sources lie inside $P_1$ and  $\frac{\partial \Phi_1}{\partial n}$ on $\partial P_2$ is uniquely defined. By Gauss theorem \citep{Gauss:1877,Backus:1996}, the spherical harmonic expansion of $\Phi_1$ on $\partial P_2$ is uniquely defined from
$\frac{\partial \Phi_1}{\partial n}$ on $\partial P_2$ and thus only contains terms related to external sources, because ${\rm supp}\,\rho_1$ is outside of $\partial P_2$.
On the other hand $\frac{\partial \Phi_1}{\partial n}~=~\frac{\partial \Phi_2}{\partial n}$ on
$\partial P_2$, and the spherical harmonic expansion of $\Phi_2$ on $\partial P_2$  has only Gauss coefficients from inner sources, because ${\rm supp}\,\rho_2$ is inside of $\partial \Omega_2$.
Because  a non-zero potential cannot at the same time  have only inner sources and only outer sources, a mutual annihilator cannot exist.
\end{proof}

\label{lastpage}

\end{document}